\documentclass[12pt]{article}
\usepackage{amsmath}
\usepackage{graphicx,psfrag,epsf}
\usepackage{enumerate}
\usepackage{psfrag,epsf}
\usepackage{url} % not crucial - just used below for the URL
\usepackage{amsmath,amssymb,amsfonts,amsthm,mathtools,mathrsfs}

\usepackage{bm,bbm}
\usepackage{color}
\usepackage{subfigure}
\usepackage{algorithm,algpseudocode}
\usepackage[symbol]{footmisc}
\usepackage{scalefnt}
\usepackage{authblk} % author and affiliations
\usepackage{multirow,centernot}
\usepackage[colorlinks=true, citecolor=blue, urlcolor=blue]{hyperref}
\usepackage{natbib}
\usepackage{dsfont}
\usepackage[title]{appendix}
\input cyracc.def

\usepackage[left=1in,top=1.1in,right=0.5in,bottom=1in]{geometry}

\theoremstyle{definition}

\newtheorem{theorem}{Theorem}[section]

\newtheorem{corollary}[theorem]{Corollary}

\newtheorem{definition}{Definition}[section]

\newtheorem{proposition}[theorem]{Proposition}
\newtheorem{remark}[theorem]{Remark}
\makeatletter
\def\@seccntformat#1{\@ifundefined{#1@cntformat}%
	{\csname the#1\endcsname\quad}%      default
	{\csname #1@cntformat\endcsname}%    enable individual control
}
\makeatother

\newif\ifShowComments
\ShowCommentstrue
\def\strutdepth{\dp\strutbox}
\def\druk#1{\strut\vadjust{\kern-\strutdepth
        {\vtop to \strutdepth{%
                \baselineskip\strutdepth\vss
                        \llap{\hbox{#1}\quad}\null}}}}

\title{\bf
Closed-form formulas for the biases of the Theil and Atkinson index estimators in Gamma distributed populations
}

\author{
\text{Roberto Vila}$^{1}$\thanks{Corresponding author: Roberto Vila, email: {rovig161@gmail.com}
}
\,\, and
\text{Helton Saulo}$^{1,2}$ 
\\
{\small $^{1}$ Department of Statistics, University of Brasilia, Brasilia, Brazil}\\
{\small $^{2}$ Department of Economics, Federal University of Pelotas, Pelotas, Brazil}\\
}
\begin{document}
	\maketitle 	
	\begin{abstract}
%This paper examines the Hoover index estimator for gamma populations, revealing bias. We derive a bias expression and propose a bias-corrected estimator, which is evaluated through Monte Carlo simulations.
%
This paper presents an analysis of the Theil and Atkinson index estimators for gamma populations, highlighting the presence of bias in both cases. Theoretical expressions for the biases are obtained, and bias-corrected estimators are constructed and evaluated through simulation studies.
	\end{abstract}
%	
%%\tableofcontents

\section{Introduction}

Measures of economic inequality are vital instruments in the quantitative examination of income distribution.  The Theil and Atkinson indexes are important measures of inequality, both of which have strong foundations in information theory and social welfare analysis, respectively; see \citet{Theil1967, Atkinson1970, Hoffmann2000}.  These indices measure the degree of inequality within a population and have been utilized in economics, software maintenance, environment, sustainability, among others; see, for example, \cite{5609637, ConceicaoGalbraith2000, DU20151080, SUN2015751, Malakar2018}

In this paper, we analyze the bias of the sample Theil $T$, Theil $L$, and Atkinson index estimators under the assumption that the population follows a gamma distribution \citep{McDonald1979}. We derive closed-form expressions for the expectations of these estimators and use them to construct bias-corrected versions. We conduct a Monte Carlo simulation to evaluate the performance of the original and proposed estimators. The simulation results show that the proposed bias-corrected estimators consistently outperform the original ones, substantially reducing the relative bias across all configurations considered. The present study is motivated by the fact that, to the best of our knowledge, no unbiased estimators for the Theil and Atkinson indexes exist, despite the fact that analogous efforts have been made to construct them for the Gini coefficient in gamma populations; see \cite{Baydil2025}.

The rest of this paper proceeds as follows. In Section~\ref{sec:02}, we define the Theil and Atkinson indexes and present their population expressions under the gamma distribution. In Section~\ref{sec:03}, we derive closed-form expectations of the sample estimators, and also present bias expressions for the considered estimators. In Section~\ref{sec:04}, we propose bias-corrected estimators and report results from a simulation study assessing the finite-sample performance of all estimators. Finally, in Section~\ref{sec:05}, we conclude the paper with final remarks.

\section{Theil and Atkinson indexes for gamma distributions}\label{sec:02}

\subsection{Theil index}

The Theil index is a statistical measure of economic inequality, quantifying the divergence from a perfectly equal distribution of income.
	\begin{definition}\label{def-1-1-1}
	The Theil $T$ index \citep{Theil1967} of a random variable $X$ with finite mean $\mathbb{E}(X)=\mu$ is defined as
	\begin{align}\label{Hoover-index}
		T_T=	
		{\mathbb{E}\left[{X\over\mu}\, \log\left({X\over\mu}\right)\right]},
	\end{align}
	and the Theil $L$ index is defined as
	\begin{align}
	T_{L}
	=
	\mathbb{E}
	\left[
	\log\left({\frac{\mu}{X}}\right)
	\right].
	\end{align}	
\end{definition}

As	$X/\mu=\lambda X/\alpha\sim\text{Gamma}(\alpha,\alpha)$ for $X\sim\text{Gamma}(\alpha,\lambda)$, and $\mathbb{E}[Z\log(Z)]=(a/b)[\psi(a)+1/a-\log(b)]$ for $Z\sim\text{Gamma}(a,b)$, the following result follows immediately.
\begin{proposition}\label{prop-Hoover-index}
	The Theil $T$ index for $X\sim\text{Gamma}(\alpha,\lambda)$ is given by
\begin{align*}
	T_T
	=
    \psi(\alpha)+{1\over\alpha}-\log(\alpha),
\end{align*}
where $\psi(x)$ is the digamma function.
\end{proposition}

As $T_L=-	\mathbb{E}
\left[
\log\left({{X}/{\mu}}\right)
\right]$
and 
$X/\mu=\lambda X/\alpha\sim\text{Gamma}(\alpha,\alpha)$ for $X\sim\text{Gamma}(\alpha,\lambda)$, from identity $\mathbb{E}[\log(Z)]=\psi(a)-\log(b)$ for $Z\sim\text{Gamma}(a,b)$, it follows that:
\begin{proposition}\label{prop-Theil-L-index}
	The Theil $L$ index for $X\sim\text{Gamma}(\alpha,\lambda)$ is given by
	\begin{align*}
		T_L
		=
		\log(\alpha)-\psi(\alpha)
		=
		{1\over\alpha}-T_T,
	\end{align*}
	where $\psi(x)$ is the digamma function.
\end{proposition}

\subsection{Atkinson index}

The Atkinson index is an indicator of income inequality, ranging from 0 to 1, used to evaluate economic disparity within a society. Increased values correspond to heightened inequality.
\begin{definition}
	The Atkinson index \citep{Atkinson1970}, denoted by $A$, of a random variable $X$ with finite mean $\mathbb{E}(X)=\mu$ is defined as
\begin{align}\label{Atkinson-index}
	A=	
	1-\exp\left\{-\mathbb{E}\left[\log\left({\mu\over X}\right)\right]\right\}
=
	1-\exp(-T_L),
\end{align}
where $T_L$ is the Theil $L$ index established in Definition \ref{def-1-1-1}.
\end{definition}

As a consequence of Proposition \ref{prop-Theil-L-index}, we have:
\begin{proposition}\label{prop-Atkinson -index}
	The Atkinson index for $X\sim\text{Gamma}(\alpha,\lambda)$ is given by
	\begin{align*}
		A
		=
		1-{\exp\{\psi(\alpha)\}\over\alpha},
	\end{align*}
	where $\psi(x)$ is the digamma function.
\end{proposition}

\section{Determining the biases of the estimators}\label{sec:03}

\subsection{Bias of the Theil  index estimators}

The next results (Theorem \ref{main-theorem} and Proposition \ref{prop-exp-L}) provide closed-form expressions for the expected value of the Theil index estimators $\widehat{T}_T$ and $\widehat{T}_L$, given by 
\begin{align}\label{gini-estimadtor-def}
	\widehat{T}_T
	=
	\dfrac{
		\displaystyle
	\sum_{i=1}^{n}
	X_{i} 
\log\left(\frac{X_i}{\overline{X}}\right)
}{\displaystyle
\sum _{i=1}^{n} X_i},
	\quad 
	n\in\mathbb{N}, 
\end{align}
and
	\begin{align}\label{def-index-L-est}
	\widehat{T}_{L}
	=
	{\frac{1}{n}}
	\sum_{i=1}^{n}
	\log\left({\frac{\overline{X}}{X_{i}}}\right),
		\quad 
	n\in\mathbb{N}, 
\end{align}	
respectively,
where $\overline{X}=(1/n)\sum_{i=1}^{n}X_i$ is the sample mean and $X_1, X_2,\ldots,X_n$ are independent, identically distributed (i.i.d.) observations from the gamma population,  thereby facilitating the calculation of the  estimator's bias in Corollary \ref{main-corollary} and Proposition \ref{main-corollary-1}. 
\begin{remark}
	Note that Theil $T$ index estimator can be written as
	\begin{align}\label{Theil-expression-alternative}
	\widehat{T}_T
	=
		\dfrac{
		\displaystyle
		\sum_{i=1}^{n}
		X_{i} 
		\log(X_i)
	}{\displaystyle
		\sum _{i=1}^{n} X_i}
		-
	\log\left(\sum_{i=1}^{n}
	X_{i}\right)
	+
	\log(n).
	\end{align}
\end{remark}

\begin{remark}
The Theil $L$ index estimator can be represented as
	\begin{align}\label{Theil-L-expression-alternative}
	\widehat{T}_{L}
	=
	\log\left(\sum_{i=1}^n X_i\right)-\log(n)
	-	
	\frac{1}{n}
	\sum_{i=1}^{n}\log(X_{i}).
\end{align}	
\end{remark}

The proofs of Theorems \ref{main-theorem} and \ref{main-theorem-1} follow similar technical steps to those in \cite{Baydil2025}.
	\begin{theorem}\label{main-theorem}
Let $X_1, X_2, \ldots$ be independent copies of $X\sim\text{Gamma}(\alpha,\lambda)$. The following holds:
\begin{align*}
	\mathbb{E}(\widehat{T}_T)
&=
	\psi(\alpha)+{1\over \alpha}
	+
	\log(n)
	-
	{1\over n\alpha}
	-
	\psi(n\alpha),
\end{align*}
where $\psi(x)$ is the digamma function.
	\end{theorem}
\begin{proof}
Utilizing the well-established identity
\begin{align}\label{fund-id}
	\int_{0}^{\infty}\exp(-z x){\rm d}x
	=
	{1\over z},
	\quad z>0,
\end{align}
with $z=\sum_{i=1}^{n} X_i$, 
%and by exploiting the independence and identical distribution of $X_1, X_2, \ldots$, 
we obtain
\begin{align}
	\mathbb{E}\left[\dfrac{\displaystyle\sum_{i=1}^{n}
		X_{i} 
		\log(X_i)}{\displaystyle\sum_{i=1}^{n} X_i}\right]
	&=
	\mathbb{E}\left[\sum_{i=1}^{n}	X_{i} 
	\log(X_i) \int_{0}^{\infty}\exp\left\{-\left(\displaystyle\sum_{i=1}^{n} X_i\right) x\right\}{\rm d}x\right]
	\nonumber
%	\\[0,2cm]
%	&=
%	\sum_{i=1}^{n}
%	\mathbb{E}\left[\, 
%	X_{i} \log(X_i) 
%	\int_{0}^{\infty}\exp(-X_i x)\exp\left\{-\left(\displaystyle\sum_{\substack{j=1\\ j\neq i}}^n X_j\right) x\right\}{\rm d}x\right]
%		\nonumber
	\\[0,2cm]
	&=
	\sum_{i=1}^{n}
	\mathbb{E}\left[\, 
	X_{i} \log(X_i) 
	\int_{0}^{\infty}\exp(-X_i x)
	\mathbb{E}\left[
	\exp\left\{-\left(\displaystyle\sum_{\substack{j=1\\ j\neq i}}^n X_j\right) x\right\}
	\right]
	{\rm d}x\right]
			\nonumber
	\\[0,2cm]
	&=
	\sum_{i=1}^{n}
	\mathbb{E}\left[\, 
	X_{i} \log(X_i) 
	\int_{0}^{\infty}\exp(-X_i x)
\mathscr{L}^{n-1}_F(x)
	{\rm d}x\right]
				\nonumber
	\\[0,2cm]
	&=
n
	\mathbb{E}\left[\, 
	X \log(X) 
	\int_{0}^{\infty}\exp(-X x)
	\mathscr{L}^{n-1}_F(x)
	{\rm d}x\right],
	\label{eq-1} 
\end{align}
where $\mathscr{L}_F(x)$ is the Laplace transform corresponding to $X\sim\text{Gamma}(\alpha,\lambda)$,  the second equality is consequence of the independence of $X_1,X_2,\ldots$ and the last equality arises from the fact that $X_1,X_2,\ldots$ are identically distributed with $X\sim\text{Gamma}(\alpha,\lambda)$.
Since the above integrands are non-negative measurable functions, Tonelli's theorem allows us to commute the order of integration, that is, the expression in \eqref{eq-1}  is
\begin{align*}
=
n
	\int_{0}^{\infty}
	\mathbb{E}\left[\, 
		X \log(X)  
		\exp(-X x)
		\right]
				\mathscr{L}^{n-1}_F(x)
	{\rm d}x.
\end{align*}
As $\mathscr{L}_F(x)=[\lambda/(x+\lambda)]^{\alpha}$, from the above identities we have
\begin{align}\label{id-1-1}
	\mathbb{E}\left[\dfrac{\displaystyle\sum_{i=1}^{n}
	X_{i} 
	\log(X_i)}{\displaystyle\sum_{i=1}^{n} X_i}\right]
=
n
\int_{0}^{\infty}
\mathbb{E}\left[\, 
X \log(X)  
\exp(-X x)
\right]
\left({\lambda\over x+\lambda}\right)^{\alpha(n-1)}
{\rm d}x.
\end{align}

On the other hand, it is simple to observe that
\begin{align*}
	\mathbb{E}\left[\, 
	X \log(X)  
	\exp(-X x)
	\right]
	=
	{\alpha\lambda^\alpha\over (x+\lambda)^{\alpha+1}} \,
	\mathbb{E}[\log(U)],
\end{align*}
where $U\sim\text{Gamma}(\alpha+1,x+\lambda)$. As $\mathbb{E}[\log(Z)]=\psi(a)-\log(b)$ for $Z\sim\text{Gamma}(a,b)$, we get
\begin{align}\label{id-2-2}
	\mathbb{E}\left[\, 
	X \log(X)  
	\exp(-X x)
	\right]
	=
	{\alpha\lambda^\alpha\over (x+\lambda)^{\alpha+1}} \,
	\left[\psi(\alpha+1)-\log(x+\lambda)\right].
\end{align}

By replacing \eqref{id-2-2} in \eqref{id-1-1}, we get
\begin{align}\label{id-3-3}
	\mathbb{E}\left[\dfrac{\displaystyle\sum_{i=1}^{n}
		X_{i} 
		\log(X_i)}{\displaystyle\sum_{i=1}^{n} X_i}\right]
%	&=
%	n\alpha\lambda^{\alpha n}
%	\int_{0}^{\infty}
%	{1\over (x+\lambda)^{\alpha n+1}} \,
%\left[\psi(\alpha+1)-\log(x+\lambda)\right]
%	{\rm d}x
%	\nonumber
%	\\[0,2cm]
%	&
	=
	\psi(\alpha+1)-\log(\lambda)-{1\over n\alpha}
%	\nonumber
%	\\[0,2cm]
%	&
	=
	\psi(\alpha)+{1\over \alpha}-\log(\lambda)-{1\over n\alpha },
\end{align}
where in last step we have applied the well-known recurrence relation:
$\psi(x+1)=\psi(x)+1/x$.

Furthermore, as $\sum_{i=1}^n X_i\sim\text{Gamma}(n\alpha,\lambda)$ and $\mathbb{E}[\log(Z)]=\psi(a)-\log(b)$ for $Z\sim\text{Gamma}(a,b)$, we have
\begin{align}\label{id-4-4}
\mathbb{E}\left[	\log\left(\sum_{i=1}^{n}
X_{i}\right)\right]
=
\psi(n\alpha)-\log(\lambda).
\end{align}

Finally, taking expectation in formula \eqref{Theil-expression-alternative} and then using identities \eqref{id-3-3} and \eqref{id-4-4}, the proof of theorem follows.
\end{proof}

\begin{remark}	
	Given the scale invariance of the Theil $T$ index estimator, $\widehat{T}_T$, it follows that its expectation, $\mathbb{E}(\widehat{T}_T)$, does not depend on the rate $\lambda$, as stated in Theorem \ref{main-theorem}.
\end{remark}

\begin{remark}
	As 
	$
		\lim_{n\to\infty}
	\left[
		\log(n)
	-
	{1/(n\alpha)}
	-
	\psi(n\alpha)
	\right]
	=
	-\log(\alpha),
	$
	from Theorem \ref{main-theorem} and Proposition \ref{prop-Hoover-index} it follows that
	$
			\lim_{n\to\infty}\mathbb{E}(\widehat{T}_T)=T_T.
	$
\end{remark}

By combining Proposition \ref{prop-Hoover-index} and Theorem \ref{main-theorem}, we have:
\begin{corollary}\label{main-corollary}
The bias of $\widehat{T}_T$ relative to $T_T$, denoted by $\text{Bias}(\widehat{T}_T,T_T)$, can be written as
\begin{align}\label{bias-Hoover-geometric}
\text{Bias}(\widehat{T}_T,T_T)
=
\log(n\alpha)
-
{1\over n\alpha}
-
\psi(n\alpha),
\end{align}
where $\psi(x)$ is the digamma function.
\end{corollary}

\begin{remark}\label{rem-1}
	Since $\log(x)-1/x-\psi(x)<0$ for all $x>0$, by Corollary \ref{main-corollary}, we have 
$
	\text{Bias}(\widehat{T}_T,T_T)<0.
$
\end{remark}

\begin{proposition}\label{prop-exp-L}
	Let $X_1, X_2, \ldots$ be independent copies of $X\sim\text{Gamma}(\alpha,\lambda)$. It holds that:
\begin{align*}
	\mathbb{E}(\widehat{T}_{L})
	=
	\psi(n\alpha)-\log(n)-\psi(\alpha)
	=
	{1\over \alpha}
	-
	{1\over n\alpha}
	-
	\mathbb{E}(\widehat{T}_T),
\end{align*}
where $\psi(x)$ is the digamma function.
\end{proposition}
\begin{proof}
	By using \eqref{Theil-L-expression-alternative} and the fact that $X_1,X_2,\ldots$ are identically distributed with $X\sim\text{Gamma}(\alpha,\lambda)$, we have
		\begin{align*}
		\mathbb{E}(\widehat{T}_{L})
		=
		\mathbb{E}\left[\log\left(\sum_{i=1}^n X_i\right)\right]
		-\log(n)
		-	
		\mathbb{E}\left[\log(X)\right].
	\end{align*}	
	Since $\sum_{i=1}^n X_i\sim\text{Gamma}(n\alpha,\lambda)$ and $\mathbb{E}[\log(Z)]=\psi(a)-\log(b)$ for $Z\sim\text{Gamma}(a,b)$, the proof follows immediately.
\end{proof}

\begin{remark}	
	Due to the scale invariance property of the Theil $L$ index estimator, $\widehat{T}_L$, it follows that its expectation, $\mathbb{E}(\widehat{T}_L)$, does not depend on the rate $\lambda$, as established in Proposition \ref{prop-exp-L}.
\end{remark}

\begin{remark}
	As $\lim_{n\to\infty}[\psi(n\alpha)-\log(n)]=\log(\alpha)$, by Propositions \ref{prop-exp-L} and \ref{prop-Theil-L-index}, its clear that
$
	\lim_{n\to\infty}\mathbb{E}(\widehat{T}_L)=T_L.
$
\end{remark}

By combining Propositions \ref{prop-Theil-L-index} and \ref{prop-exp-L}, we have:
\begin{proposition}\label{main-corollary-1}
	The bias of $\widehat{T}_L$ relative to $T_L$, denoted by $\text{Bias}(\widehat{T}_L,T_L)$, can be written as
\begin{align*}
	\text{Bias}(\widehat{T}_L,T_L)
	=
	\psi(n\alpha)-\log(n\alpha)
	=
	-\text{Bias}(\widehat{T}_T,T_T)
	-
	{1\over n\alpha},
\end{align*}
where $\psi(x)$ is the digamma function.
\end{proposition}

\begin{remark}
	Since $\psi(x)-\log(x)<0$ for all $x>0$, by Proposition \ref{main-corollary-1}, we obtain
	$
		\text{Bias}(\widehat{T}_L,T_L)<0,
	$
	and consequently,
	$
	\text{Bias}(\widehat{T}_T,T_T)>	-
	{1/(n\alpha)}.
	$
\end{remark}

\subsection{Bias of the Atkinson index estimator}

Theorem \ref{main-theorem-1} provides a closed-form expression for the expected value of the Atkinson index estimator $\widehat{A}$, defined as 
\begin{align}\label{gini-estimadtor-def-1}
	\widehat{A}
	=
	1
	-
	\dfrac{
		\displaystyle
		\left(
		\prod_{i=1}^{n}
		X_{i} 
		\right)^{1/n}
	}{\displaystyle
		{1\over n} \sum _{i=1}^{n} X_i},
	\quad 
	n\in\mathbb{N}, 
\end{align}
where 
%$\overline{X}=(1/n)\sum_{i=1}^{n}X_i$ is the sample mean and 
$X_1, X_2,\ldots,X_n$ are i.i.d. observations from the gamma population,  thereby facilitating bias computation in Corollary \ref{main-corollary-11}.

\begin{remark}\label{relationship-A-TL}
	Observe that the Atkinson  index estimator can be expressed as
$
	\widehat{A}
	=
	1-\exp(\widehat{T}_L),
$
where $\widehat{T}_L$ is the Theil $L$ index estimator defined in \eqref{def-index-L-est}.
\end{remark}

Despite the relationship between $\widehat{A}$ and $\widehat{T}_L$ provided by Remark \ref{relationship-A-TL}, and the established expectation of $\widehat{T}_L$ (see Proposition \ref{prop-exp-L}), note that determining the expected value of $\widehat{A}$ presents significant challenges.
\begin{theorem}\label{main-theorem-1}
	Let $X_1, X_2, \ldots$ be independent copies of $X\sim\text{Gamma}(\alpha,\lambda)$. We have the following result:
	\begin{align*}
		\mathbb{E}(\widehat{A})
		&=
		1-{\Gamma^n(\alpha+{1\over n})\over\alpha\Gamma^n(\alpha)},
	\end{align*}
	with $\Gamma(x)$ being the (complete) gamma function.
\end{theorem}
\begin{proof}
By using the identity \eqref{fund-id} with  $z=\sum_{i=1}^{n} X_i$, we have
\begin{align}\label{id-exp-A}
	\mathbb{E}(\widehat{A})
	&=
	1-\mathbb{E}\left[		
	\left(
	\prod_{i=1}^{n}
	X_{i} 
	\right)^{1/n}
	\int_{0}^{\infty}
	\exp\left\{-{1\over n}\left(\displaystyle\sum_{i=1}^{n} X_i\right) x\right\}{\rm d}x
	\right]
	\nonumber
	\\[0,2cm]
	&=
1-\mathbb{E}\left[	
\int_{0}^{\infty}
	\prod_{i=1}^{n}
X_{i}^{1/n}
\exp\left(-{1\over n}\, X_i x\right) 
{\rm d}x
	\right]
	\nonumber
		\\[0,2cm]
%	&
%	=
%	1-
%		\int_{0}^{\infty}
%			\prod_{i=1}^{n}
%	\mathbb{E}\left[	
%	X_{i}^{1/n}
%	\exp\left(-{1\over n}\, X_i x\right) 
%	\right]
%		{\rm d}x
&=
		1-
		\int_{0}^{\infty}
		\mathbb{E}^n\left[	
		X^{1/n}
		\exp\left(-{1\over n}\, X x\right) 
		\right]
		{\rm d}x,
\end{align}
where, as in the proof of Theorem \ref{main-theorem}, the i.i.d. nature of the involved random variables, combined with the application of Tonelli's theorem, provide the foundation for justifying the preceding steps.

Straightforward calculations yield   
\begin{align*}
\mathbb{E}\left[	
X^{1/n}
\exp\left(-{1\over n}\, X x\right) 
\right]
=
{\lambda^\alpha \Gamma(\alpha+{1\over n})\over\Gamma(\alpha)({1\over n} x+\lambda)^{\alpha+{1\over n}}},
\end{align*}
which implies 
\begin{align} \label{id-exp-A-1}
		\int_{0}^{\infty}
\mathbb{E}^n\left[	
X^{1/n}
\exp\left(-{1\over n}\, X x\right) 
\right]
{\rm d}x
=
{\Gamma^n(\alpha+{1\over n})\over\alpha\Gamma^n(\alpha)}.
\end{align}
By combining \eqref{id-exp-A} and \eqref{id-exp-A-1} yields the desired result, thus completing the proof.
\end{proof}

\begin{remark}	
The scale invariance property of $\widehat{A}$ implies that its expectation $\mathbb{E}(\widehat{A})$ is unaffected by the rate $\lambda$, as shown in Theorem \ref{main-theorem-1}
\end{remark}

\begin{remark}
As 
$
	\lim_{n\to\infty}
{\Gamma^n(\alpha+{1/n})/\Gamma^n(\alpha)}
=
\exp\{\psi(\alpha)\},
$
from Theorem \ref{main-theorem-1} and Proposition \ref{prop-Atkinson -index} it follows that
$
		\lim_{n\to\infty}\mathbb{E}(\widehat{A})=A.
$
\end{remark}

By combining Proposition \ref{prop-Atkinson -index} and Theorem \ref{main-theorem-1}, we have:
\begin{corollary}\label{main-corollary-11}
	The bias of $\widehat{A}$ relative to $A$, denoted by $\text{Bias}(\widehat{A},A)$, can be written as
	\begin{align}\label{bias-A}
		\text{Bias}(\widehat{A},A)
		=
		{1\over \alpha}
		\left[
		\exp\{\psi(\alpha)\}
		-
		{\Gamma^n(\alpha+{1\over n})\over\Gamma^n(\alpha)}
		\right],
	\end{align}
	where $\Gamma(x)$ and $\psi(x)$ are the gamma and digamma functions, respectively.
\end{corollary}

%\newpage

\section{Illustrative simulation study}\label{sec:04}

Based on the analytical bias expressions derived in Section~\ref{sec:03}, we propose bias-corrected versions of the estimators for the Theil $T$, Theil $L$, and Atkinson indexes, given by
\eqref{gini-estimadtor-def},
\eqref{def-index-L-est}, and
\eqref{gini-estimadtor-def-1}, respectively. Let $X_1, X_2, \dots, X_n$ be an i.i.d. sample from a Gamma$(\alpha, \lambda)$ distribution, and denote by $\widehat{\alpha}$ the maximum likelihood estimator of the shape parameter $\alpha$. Then, the bias-corrected estimators are given as follows.
\begin{itemize}
    \item {Theil $T$ index:}
    \begin{equation}
    \widehat{T}_T^{\,\text{corr}} =
    \left[
        \frac{\sum_{i=1}^n X_i \log(X_i)}{\sum_{i=1}^n X_i}
        - \log\left( \frac{1}{n} \sum_{i=1}^n X_i \right)
    \right]
    - \left[
        \log(n \widehat{\alpha}) - \frac{1}{n \widehat{\alpha}} - \psi(n \widehat{\alpha})
    \right],
    \label{eq:theil_T_corr}
    \end{equation}.

    \item {Theil $L$ index:}
    \begin{equation}
    \widehat{T}_L^{\,\text{corr}} =
    \left[
        \log\left( \frac{1}{n} \sum_{i=1}^n X_i \right)
        - \frac{1}{n} \sum_{i=1}^n \log X_i
    \right]
    - \left[
        \psi(n \widehat{\alpha}) - \log(n \widehat{\alpha})
    \right],
    \label{eq:theil_L_corr}
    \end{equation}

    \item {Atkinson index:}
    \begin{equation}
    \widehat{A}^{\,\text{corr}} =
    \left[
        1 - \frac{\left( \prod_{i=1}^n X_i \right)^{1/n}}{\frac{1}{n} \sum_{i=1}^n X_i}
    \right]
    - \frac{1}{\widehat{\alpha}} \left[
        \exp\{\psi(\widehat{\alpha})\}
        - \frac{\Gamma^n\left( \widehat{\alpha} + \frac{1}{n} \right)}{\Gamma^n(\widehat{\alpha})}
    \right],
    \label{eq:atkinson_corr}
    \end{equation}
\end{itemize}
where $\Gamma(x)$ and $\psi(x)$ are the gamma and digamma functions, respectively.

We assess the performance of the original and bias-corrected estimators for the Theil $T$, Theil $L$, and Atkinson indexes, by means of a Monte Carlo simulation study. Independent and identically distributed samples were generated from the Gamma distribution with varying shape parameters $\alpha \in \{0.1, 0.5, 1.5, 2.0\}$, rate parameter $\lambda=1.0$, and sample sizes $n \in \{10, 20, 50, 100, 200\}$. For each scenario, 1,000 replications were performed. The steps of the Monte Carlo simulation study are described in Algorithm 1.

\begin{algorithm}[!ht]
\caption{Monte Carlo simulation for bias-corrected Theil and Atkinson index estimators under the Gamma model.}
\label{algorithm:gamma}
\begin{algorithmic}[1]
\State \textbf{Input:} Number of simulations $N_{\text{sim}} = 1000$; sample sizes $n \in \{10, 20, 50, 100, 200\}$; shape parameters $\alpha \in \{0.1, 0.5, 1.5, 2.0\}$; rate parameter $\lambda=1.0$.
\State \textbf{Output:} Bias and MSE of original and bias-corrected estimators: $\widehat{T}_T$, $\widehat{T}_T^{\text{corr}}$, $\widehat{T}_L$, $\widehat{T}_L^{\text{corr}}$, $\widehat{A}$, and $\widehat{A}^{\text{corr}}$.

\For{each shape parameter $\alpha$}
  \State Compute the true population values $T_T$, $T_L$, and $A$ using Propositions~2.1--2.3.
  \For{each sample size $n$}
    \For{each simulation run $s = 1, \dots, N_{\text{sim}}$}
        \State \textbf{Step 1: Generate data}
        \State Simulate $X_1, \dots, X_n \sim \text{Gamma}(\alpha, \lambda = \alpha)$.
        \State Estimate $\widehat{\alpha}$ via maximum likelihood.

        \State \textbf{Step 2: Compute estimators}
        \State Compute $\widehat{T}_T$ using Equation~\eqref{gini-estimadtor-def}.
        \State Compute $\widehat{T}_T^{\text{corr}}$ using Equation~\eqref{eq:theil_T_corr}.
        \State Compute $\widehat{T}_L$ using Equation~\eqref{def-index-L-est}.
        \State Compute $\widehat{T}_L^{\text{corr}}$ using Equation~\eqref{eq:theil_L_corr}.
        \State Compute $\widehat{A}$ using Equation~\eqref{gini-estimadtor-def-1}.
        \State Compute $\widehat{A}^{\text{corr}}$ using Equation~\eqref{eq:atkinson_corr}.
    \EndFor

    \State \textbf{Step 3: Compute Monte Carlo summaries}
    \For{each estimator $E \in \{\widehat{T}_T, \widehat{T}_T^{\text{corr}}, \widehat{T}_L, \widehat{T}_L^{\text{corr}}, \widehat{A}, \widehat{A}^{\text{corr}}\}$}
        \State Compute empirical bias:
        \[
        \widehat{\text{Relative Bias}}(E) = \frac{1}{N_{\text{sim}}} \sum_{k=1}^{N_{\text{sim}}} \frac{E^{(k)} - E_{\text{true}}}{E_{\text{true}}},
        \]
        \State Compute mean squared error:
        \[
        \widehat{\text{MSE}}(E) = \frac{1}{N_{\text{sim}}} \sum_{k=1}^{N_{\text{sim}}} \big(E^{(k)} - E_{\text{true}}\big)^2,
        \]
        where $E_{\text{true}} \in \{T_T, T_L, A\}$ is the corresponding true value.
    \EndFor
  \EndFor
\EndFor
\State \textbf{Return:} Summary with bias and MSE for all estimators across all settings.
\end{algorithmic}
\end{algorithm}

Figures~\ref{fig:theilT}, \ref{fig:theilL}, and \ref{fig:atkinson} summarize the Monte Carlo simulation results for the original and bias-corrected estimators of the Theil $T$, Theil $L$, and Atkinson indexes, respectively. Each figure reports the relative bias and mean squared error (MSE) of the estimators as functions of the sample size (with fixed shape parameter $\alpha = 1.5$) and as functions of the shape parameter $\alpha$ (with fixed sample size $n = 20$). From these figures, we observe that the bias-corrected estimators exhibit markedly lower bias than their original counterparts, especially for small samples. Moreover, the proposed corrections prove effective across the entire range of $\alpha$, maintaining low relative bias. Fnally, we observe that, in general, both estimators exhibit similar MSE values across the configurations considered.

\begin{figure}[!ht]
\centering
\includegraphics[width=0.48\textwidth]{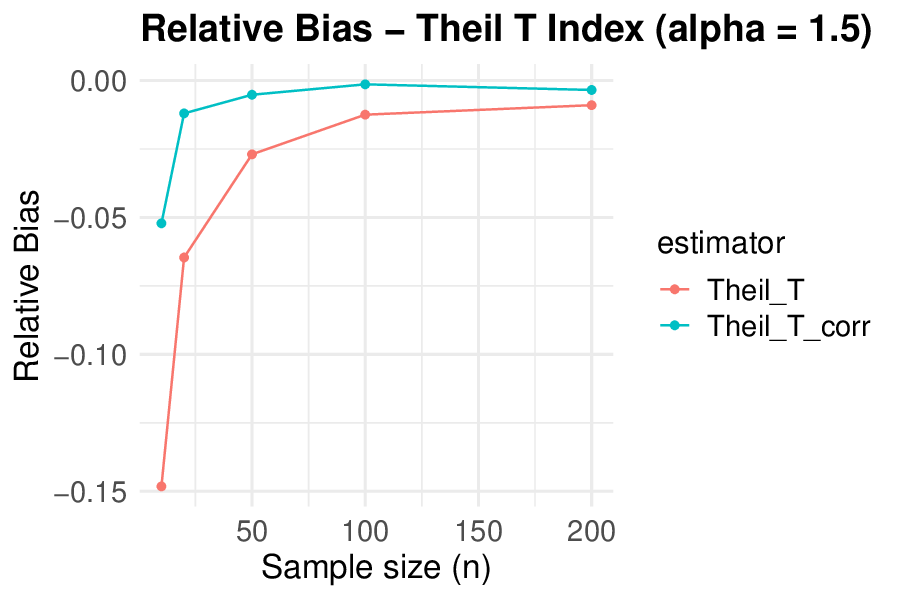}
\includegraphics[width=0.48\textwidth]{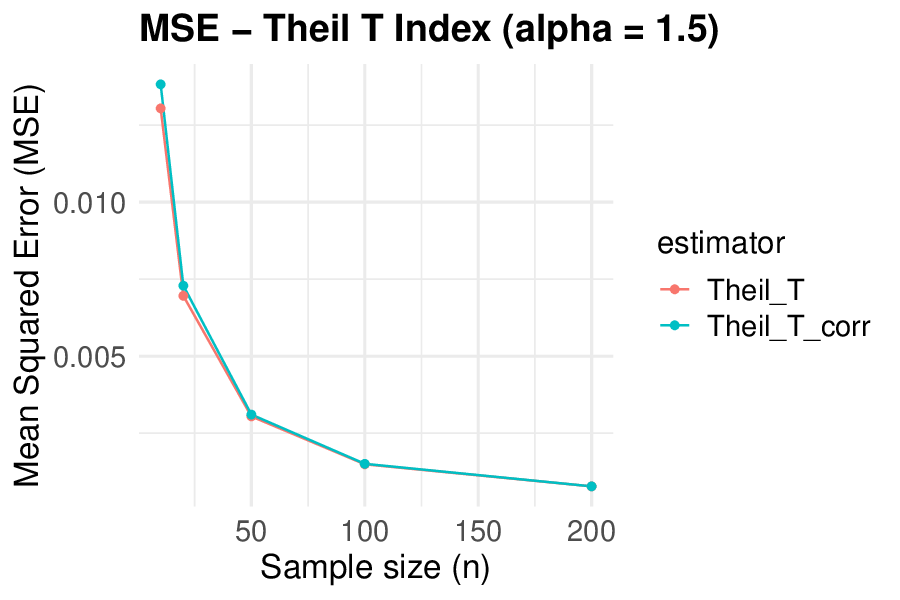}

\vspace{0.4cm}

\includegraphics[width=0.48\textwidth]{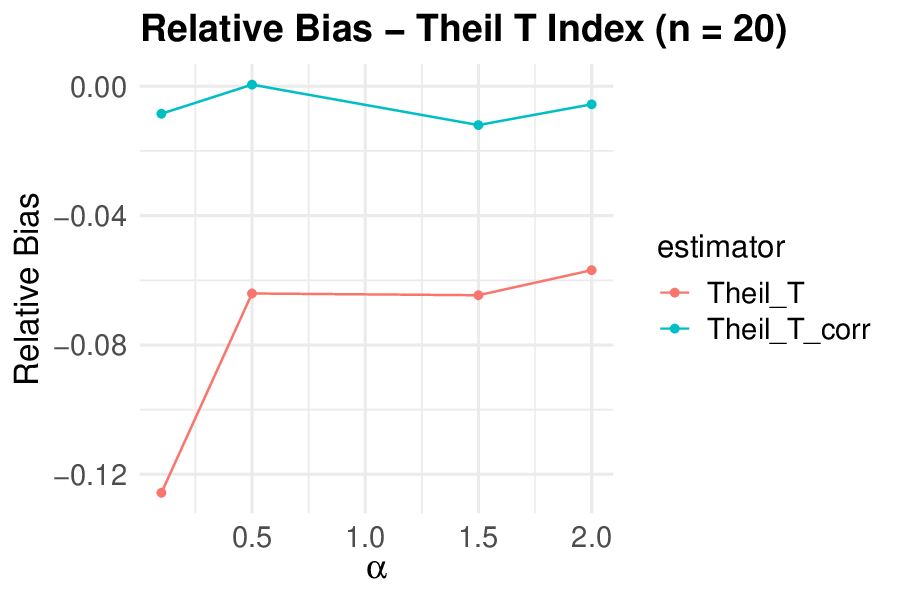}
\includegraphics[width=0.48\textwidth]{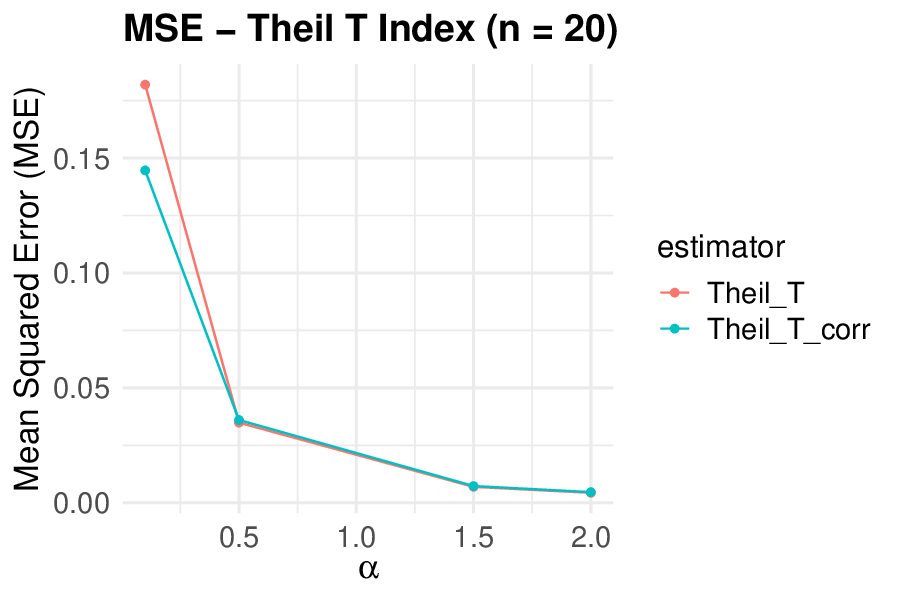}

\caption{Monte Carlo results for the Theil $T$ index estimators (original and bias-corrected). Top: relative bias and MSE as functions of sample size (\(\alpha = 1.5\)). Bottom: relative bias and MSE as functions of shape parameter \(\alpha\) (\(n = 20\)).}
\label{fig:theilT}
\end{figure}

\begin{figure}[!ht]
\centering
\includegraphics[width=0.48\textwidth]{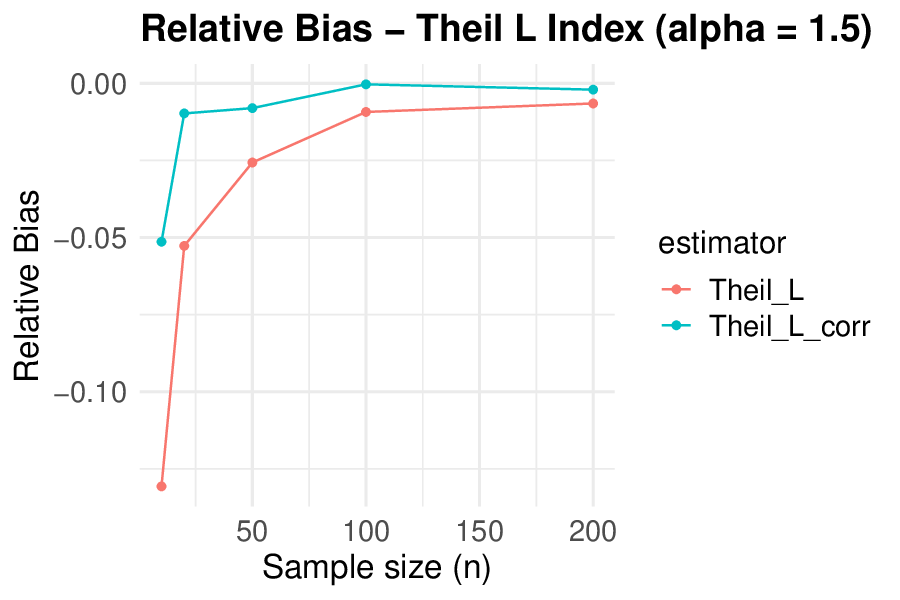}
\includegraphics[width=0.48\textwidth]{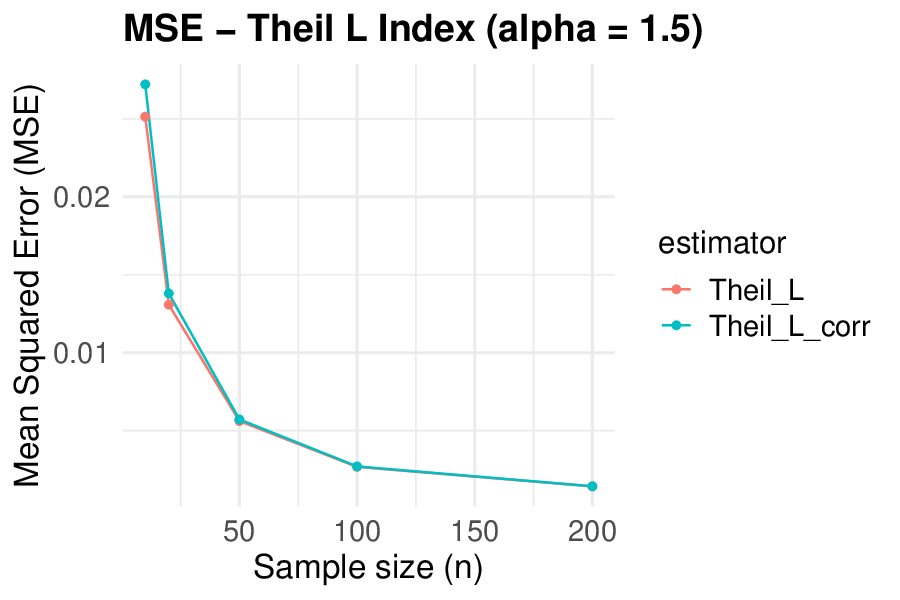}

\vspace{0.4cm}

\includegraphics[width=0.48\textwidth]{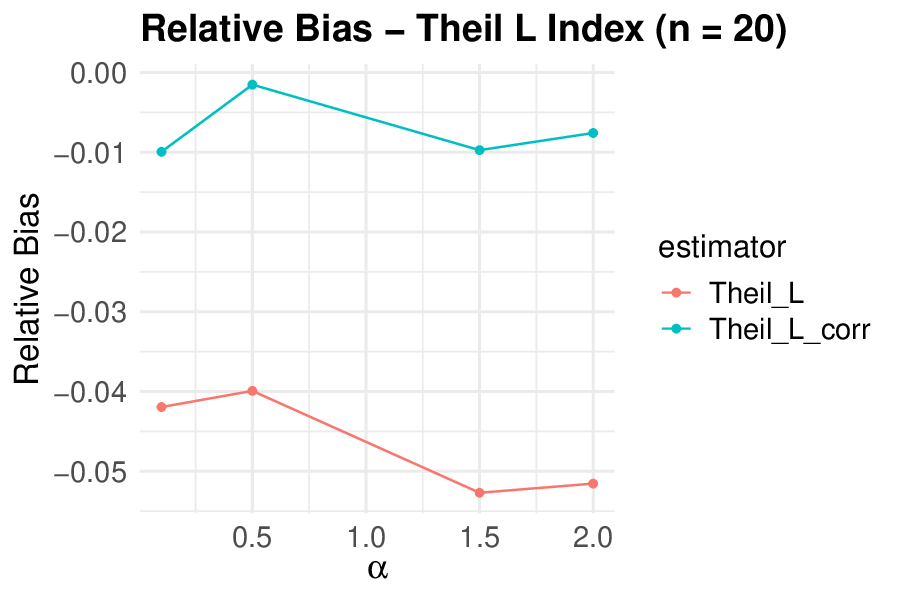}
\includegraphics[width=0.48\textwidth]{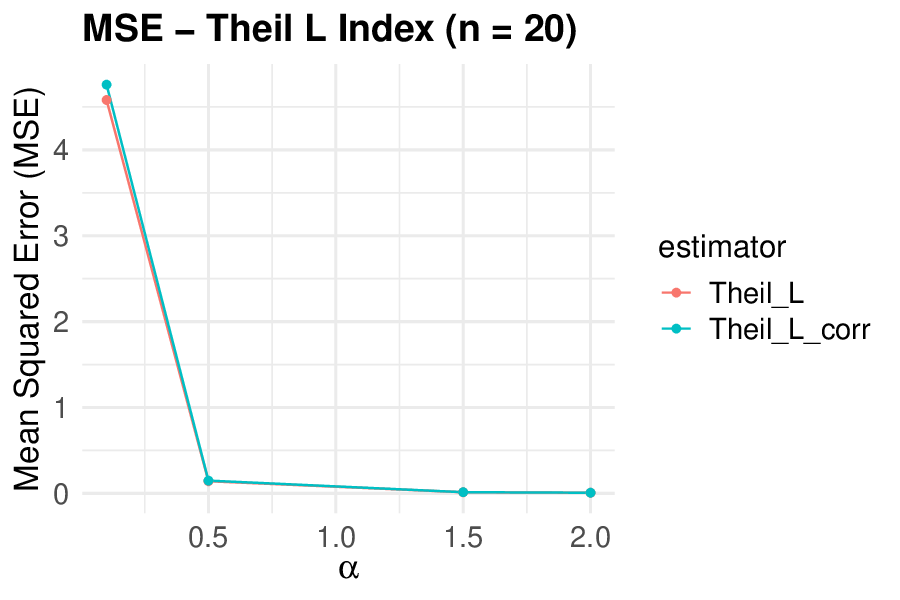}

\caption{Monte Carlo results for the Theil $L$ index estimators (original and bias-corrected). Top: relative bias and MSE as functions of sample size (\(\alpha = 1.5\)). Bottom: relative bias and MSE as functions of shape parameter \(\alpha\) (\(n = 20\)).}
\label{fig:theilL}
\end{figure}

\begin{figure}[!ht]
\centering
\includegraphics[width=0.48\textwidth]{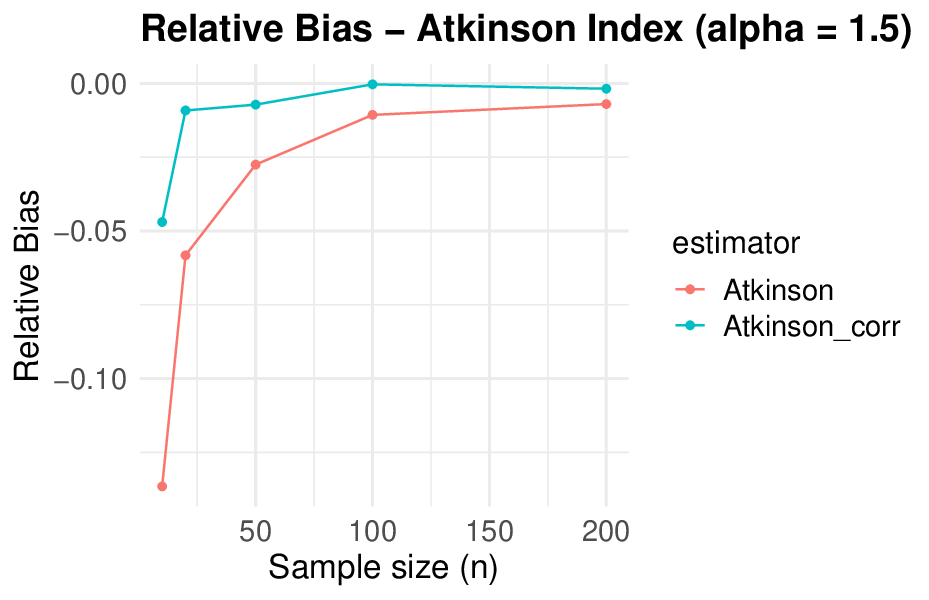}\hspace{0.4cm}
\includegraphics[width=0.48\textwidth]{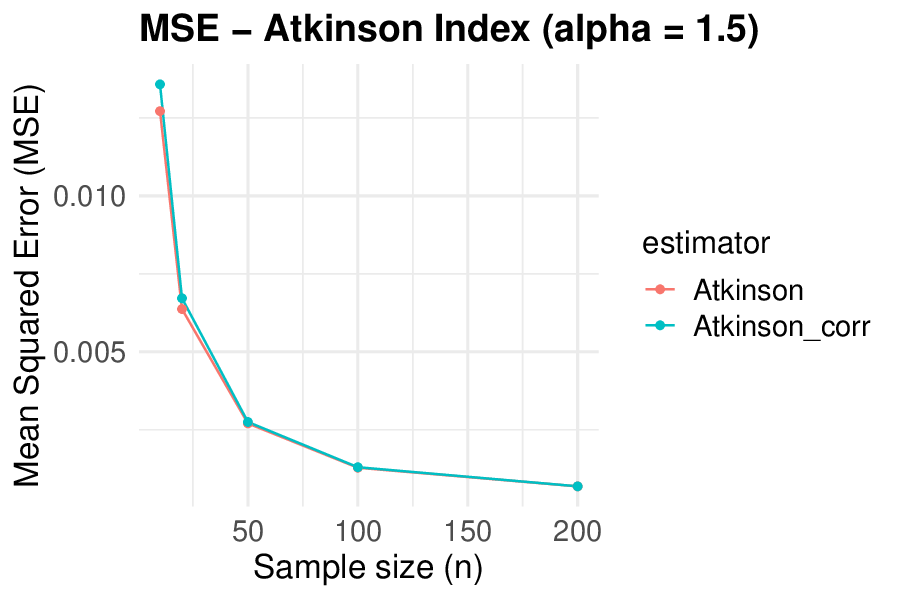}

\vspace{0.4cm}

\includegraphics[width=0.48\textwidth]{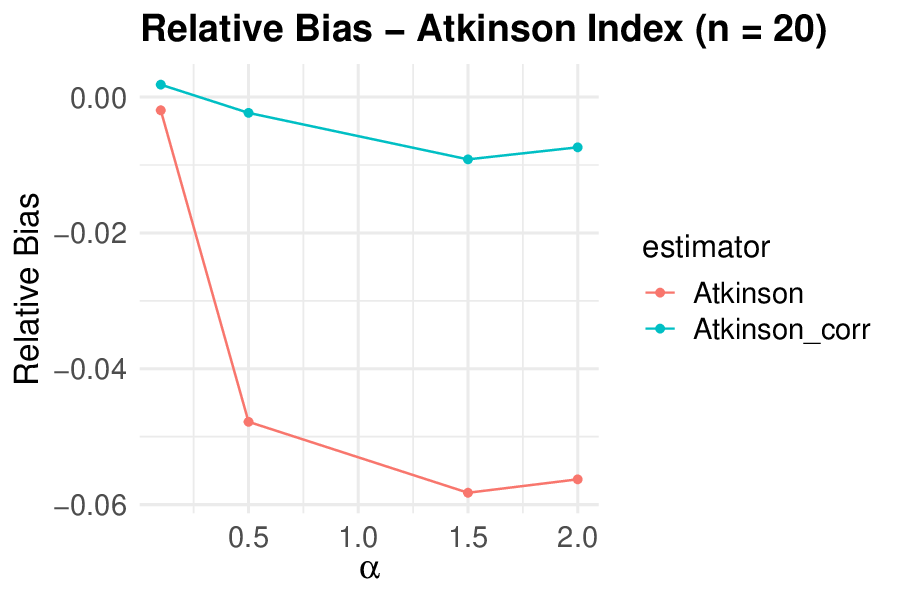}
\includegraphics[width=0.48\textwidth]{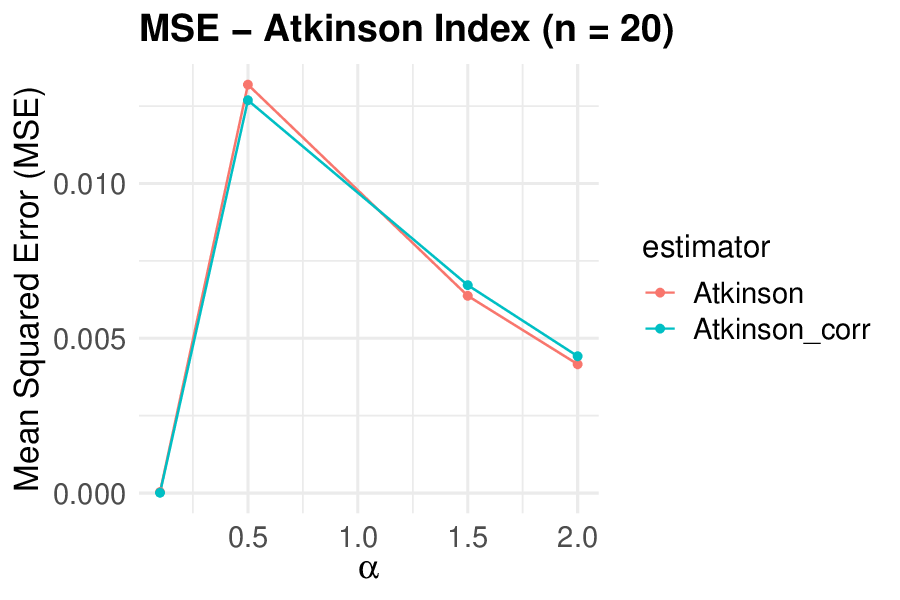}

\caption{Monte Carlo results for the Atkinson index estimators (original and bias-corrected). Top: relative bias and MSE as functions of sample size (\(\alpha = 1.5\)). Bottom: relative bias and MSE as functions of shape parameter \(\alpha\) (\(n = 20\)).}
\label{fig:atkinson}
\end{figure}

\section{Concluding remarks}\label{sec:05}

In this paper, we have investigated the bias in the estimation of the Theil $T$, Theil $L$, and Atkinson indexes under gamma-distributed populations. We have derived closed-form expressions for the biases of these estimators and used them to construct explicit bias-corrected versions. An illustrative Monte Carlo simulation study was conducted to evaluate the finite-sample performance of the original and proposed bias-corrected estimators. The results demonstrate that the corrected estimators substantially reduce the relative bias across all configurations considered, offering significant improvements over the original estimators while maintaining similar levels of mean squared error. As part of future research, it will be of interest to extend the study to multivariate Theil $T$, Theil $L$, and Atkinson indexes. Furthermore, alternative correction methods or the analysis of gamma mixture models can be developed. Work on these problems is currently in progress and we hope to report these findings in future.

%\clearpage

%\clearpage
%
%
%	\paragraph*{Acknowledgements}
%The research was supported in part by CNPq and CAPES grants from the Brazilian government.
%	
%	\paragraph*{Disclosure statement}
%	There are no conflicts of interest to disclose.

	%%%%%%%%%%%%%%%%%%%%%%%%%%%%%%%%%%%%%%%%%%%%%%%%%%%%%%%%%%%%%

\clearpage

\paragraph*{Acknowledgements}
The research was supported in part by CNPq and CAPES grants from the Brazilian government.

\paragraph*{Disclosure statement}
There are no conflicts of interest to disclose.

%%%%%%%%%%%%%%%%%%%%%%%%%%%%%%%%%%%%%%%%%%%%%%%%%%%%%%%%%%%%%

%\bibliographystyle{unsrt}
%\bibliographystyle{apalike}
%\bibliography{references}

%\appendix
%\section{Appendix}

%\typeout{}
%\bibliography{ref}

\end{document}